\documentclass[11pt,a4paper,reqno]{amsart}
\usepackage[latin1]{inputenc}
\usepackage{amsthm,amsmath,amsfonts,amssymb,amsxtra,appendix,bookmark,dsfont,bm,mathrsfs,amstext,amsopn,mathrsfs,mathtools,comment,cite,hyperref,color}
\usepackage{tikz}\usetikzlibrary{arrows}
\usepackage{macros_art}
\usepackage{macros}
\usepackage{bbm} \newcommand{\dis}{\mathbbmss{d}}

\usepackage{algorithm}

\newcommand{\vde}{W}
\newcommand{\rohk}{\ro}

\title[HK theorems]{Hohenberg-Kohn theorems \\ for interactions, spin and temperature}
\author[L. Garrigue]{Louis Garrigue}
\address{CEREMADE, Universit\'e Paris-Dauphine, PSL Research University, F-75016 Paris, France} 
\email{garrigue@ceremade.dauphine.fr}

\begin{document} 

%
%

\begin{abstract} 
We prove Hohenberg-Kohn theorems for several models of quantum mechanics. First, we show that the pair correlation function of any ground state contains the information of the interactions and of the external potentials. Then, in the presence of the Zeeman interaction, a strong constraint on external fields is derived for systems having the same ground state densities and magnetizations. Moreover, we provide a counterexample in a setting involving non-local potentials. Next, we prove that the density and the entropy of a ground state contain the information of both the imposed external potential and temperature. Eventually, we conclude that at positive temperature, Hohenberg-Kohn theorems generically hold.
\end{abstract}
\date{\today}

\maketitle
\setcounter{tocdepth}{1} 

\section{Introduction} 

The ground state of a many-body quantum system at equilibrium is a quantity involving a large number of variables, and one cannot measure it directly. Therefore, it is natural to ask which simple and measurable quantity is sufficient to know to get all the information of a system. Then one can work only with this relevant reduced information as in Density Functional Theory, which is one of the most successful methods in quantum physics and chemistry to simulate matter at the microscopic scale \cite{HohKoh64,KohSha65,Burke12,Jones15}. A famous result of Hohenberg and Kohn \cite{HohKoh64} from 1964 shows that, at equilibrium, the ground state density of a system of quantum electrons contains all the information of the external electric potential. This implies that any physical quantity is a functional of this density and justifies Density Functional Theory. Then, many articles were devoted to extend this initial work to other configurations, because those conceptual results provide fundamental insights on the structure of quantum mechanical models. The main goal is to establish a bijective matching between an external imposed field, and a ground state internal reduced density.

In chronological order but not exhaustively, Mermin extended the theorem to fixed positive temperature \cite{Mermin65} (Thermal DFT), Barth and Hedin looked at Hamiltonians having a Zeeman term \cite{BarHed72} (Spin DFT), Gilbert considered non-local potentials \cite{Gilbert75} (Matrix DFT), Vignale and Rasolt treated Pauli Hamiltonians without the Zeeman term \cite{VigRas87} (Current DFT), and Siedentop, M\"uller and Ziesche considered pair potentials \cite{SieMul81,Ziesche94} (Pair DFT). In those works, the authors conjectured corresponding Hohenberg-Kohn theorems, that is some reduced ground state densities contain the information of external potentials applied to the systems. Nevertheless, the provided proofs turn out to be incomplete and first counterexamples were found by Capelle and Vignale in models dealing with magnetic fields \cite{CapVig01,CapVig02} (Spin DFT, Spin-Current DFT and Superconducting DFT at zero temperature), invalidating such general Hohenberg-Kohn properties. In two recent articles \cite{Garrigue18,Garrigue19}, we proved a unique continuation property necessary in the final step of the rigorous proof of the original Hohenberg-Kohn theorem. But our work does not cover much of the systems mentionned above.

In this article we further analyze Hohenberg-Kohn theorems for interactions, spin, non-local potentials and temperature, and complete the proofs present in the litterature when possible. First, we comment on the original Hohenberg-Kohn property. In Theorem \ref{types}, we show a Hohenberg-Kohn theorem for interactions, indicating that pair correlations of any ground state are sufficient to deduce the interactions between particles, in general settings containing several types of particles. Next in Theorem \ref{spinhk}, we prove a partial Hohenberg-Kohn type result for Spin DFT, that is, a strong constraint on external fields when ground state one-body densities are equal, and we provide a counterexample to the Hohenberg-Kohn theorem in Matrix DFT. Then, we give a rigorous proof of the Hohenberg-Kohn theorem in Thermal DFT in Theorem \ref{hkt}, extending the existing statements by showing that the ground state entropy and the one-body density contain the information of both the temperature and the external electric potential. Eventually, we show that at positive temperatures, Hohenberg-Kohn theorems \apo{generically} hold, in particular in the classical case.

\subsection*{Acknowledgement}
I warmly thank Mathieu Lewin, my PhD advisor, for having supervized me during this work. This project has received funding from the European Research Council (ERC) under the European Union's Horizon 2020 research and innovation programme (grant agreement MDFT No 725528), and from the Allocation Moniteur Normalien.

\section{Standard setting and interactions}

The original Hohenberg-Kohn theorem initiated many works, like \cite{HohKoh64,Lieb83b,Mezey99,PinBokLud07,Garrigue19} for instance. In this section we present weaker assumptions under which it holds. Most of them can be applied to other Hohenberg-Kohn theorems presented in sections below. We denote by $d$ the dimension of the one-body space $\R^d$, and by $N$ the number of particles. 
Let us consider the interacting Hamiltonian
\begin{align}\label{opcl}
H^N(v) \df \sum_{i=1}^N -\Delta_i + \sum_{i=1}^N v(x_i) + \sum_{1 \le i \sle j \le N} w(x_i-x_j),
\end{align}
where $v$ and $w$ are respectively the external and interacting potentials. Let $q \in \N$ be the spin number. The one-body density of a wavefunction $\p \in L^2(\R^{dN},\C^{q^N})$, simply called the density hereafter, is defined by
\begin{multline*}
\ro_{\p}(x)  \df \sum_{s \in \acs{1,\dots,q}^N} \sum_{i=1}^N \int_{\R^{d(N-1)}} \ab{\p^s}^2(x_1,\dots,x_{i-1},x,x_{i+1},\dots,x_N) \\
\times \d x_1 \cdots \d x_{i-1} \d x_{i+1} \cdots \d x_N.
\end{multline*}
The set of square-integrable $N$-particle antisymmetric wavefunctions will be denoted by $L^2_{\tx{a}}(\R^{dN}) \df \bigwedge_{i=1}^N L^2(\R^d)$. 
In the paper we will write $E_i$, $i \in \acs{1,2}$ for ground state energies of the corresponding models we are treating. For instance here, $E_i \df \ps{\p_i, H(v_i) \p_i}$.

\subsection{Standard Hohenberg-Kohn}

The usual assumption in the Hohenberg-Kohn theorem is the equality of ground state densities, $\ro_{\p_1} = \ro_{\p_2}$ almost everywhere, where $\p_1$ and $\p_2$ are ground states produced by two electric potentials $v_1$ and $v_2$. In the case of Coulomb systems of molecules, this can be replaced by $\ro_{\p_1} = \ro_{\p_2}$ in a ball, by real analyticity \cite{Mezey99,FouHofOst02a,FouHofOst02c,Jecko10}. Nevertheless, we remark that in the general case, the proof only requires the following constraint \eqref{assump}, which we present together with the proof of the Hohenberg-Kohn theorem \cite{HohKoh64,Lieb63b,PinBokLud07,Garrigue18,Garrigue19}, for completeness. We will denote by $f_+ = \max(f,0)$ and $f_- = \max(-f,0)$ the positive and negative parts of a function $f = f_+-f_-$. Let $\Omega$ be a connected open set, we work on $\Omega^N$ with any boundary condition for the Laplacian. We call $\cQ\bpa{\sum_{i=1}^N -\Delta_i}$ the corresponding form domain. For $v \in L^{d/2}\loc(\R^d)$ such that $v_- \in (L^{d/2} + L^{\ii})(\R^d)$, we also denote by $H^N(v)$ the Friedrich extension of the operator in \eqref{opcl}, whose form domain is
 \begin{align*}
\cQ \bpa{ H^N(v) } & = \acs{ \p \in \cQ\pa{\smallsum_{i=1}^N -\Delta_i}  \bigst \mediumint_{\R^d} v_+ \ro_{\p} \sle + \ii  }.
 \end{align*}

\begin{theorem}[Hohenberg-Kohn theorem]\label{clHK}
	Let $\Omega \subset \R^d$ be an open and connected set, and choose any boundary condition. Let $p > \max(2d/3,2)$, let $v_1,v_2 \in L^p\loc(\Omega)$ having $(v_1)_-,(v_2)_-\in (L^p+L^{\ii})(\Omega)$, and $w \in (L^p+L^{\ii})(\R^d)$ such that $H^N(v_1)$ and $H^N(v_2)$ have at least one ground state each, we respectively denote by $\p_1$ and $\p_2$ one of the ground states of $H^N(v_1)$ and $H^N(v_2)$, and we assume that $\int (v_1)_+ \ro_{\p_2}$ and $\int (v_2)_+ \ro_{\p_1}$ are finite. If 
\begin{align}\label{assump}
\int_{\Omega}(v_1-v_2) (\ro_{\p_1}-\ro_{\p_2}) = 0,
\end{align}
then $v_1 = v_2 + (E_1-E_2)/N$.
\end{theorem}
\begin{proof}
From the conditions imposed on the potentials, $\p_2 \in \cQ\pa{H^N(v_1)}$ and $\p_1 \in \cQ\pa{H^N(v_2)}$. Using the definition of the ground state energy, we have
 \begin{align*}
E_1 \leq \ps{\p_2,H^N(v_1) \p_2} = E_2 + \int (v_1 - v_2) \ro_{\p_2},
 \end{align*}
which can be written
\begin{align}\label{e1}
E_1 - E_2 & \leq \int (v_1- v_2) \ro_{\p_2}.
\end{align}
Exchanging labels $1$ and $2$, we also have
\begin{align}\label{e2}
E_2 - E_1 & \leq \int (v_2- v_1) \ro_{\p_1}.
\end{align}
Now using the hypothesis \eqref{assump}, we get
\begin{align*}
E_1 - E_2 = \int (v_1- v_2) \ro_{\p_1} = \int (v_1-v_2) \ro_{\p_2},
\end{align*}
and thus $\ps{\p_2,H^N(v_1) \p_2} = E_1$. Consequently, $\p_2$ verifies Schr\"odinger's equation for $H^N(v_1)$, that is $H^N(v_1) \p_2 = E_1 \p_2$. Taking the difference with Schr\"odinger's equation verified by $\p_2$ for $H^N(v_2)$, we have
\begin{align*}
	\pa{E_2 - E_1 + \sum_{i=1}^N (v_1-v_2)(x_i)} \p_2 = 0.
\end{align*}
Since Schr\"odinger's equation is verified by $\p_2$, it does not vanish on sets of positive measure by unique continuation \cite{Garrigue18,Garrigue19}, which is a local property, and we get
\begin{align}\label{eqhk}
 E_2 - E_1 + \sum_{i=1}^N (v_1-v_2)(x_i) = 0,
\end{align}
a.e. in $\Omega^N$. Integrating this equation on $(x_i)_{2\le i \le N} \in A^{N-1}$ where $A \subset \Omega$ is bounded, we obtain that there is a constant $c$ such that $v_1 = v_2 + c$ almost everywhere on $\Omega$. Eventually, $c=(E_1-E_2)/N$ by using \eqref{eqhk} once again.
\end{proof}

This proof relies on a unique continuation result \cite{Garrigue19}, itself based on a Carleman inequality, and its use yields the condition $p > \max(2d/3,2)$. In the quantum case, the Laplace operator forces minimizers to be spread in the whole space by unique continuation, which is linked to the Heisenberg principle. Note that, on the contrary, in the classical case minimizers of the energy at zero temperature have a very small support (they concentrate on the minimizers of the energy), and therefore they provide almost no information on
the potential.

The hypothesis \eqref{assump} is a much weaker assumption than $\ro_{\p_1} = \ro_{\p_2}$. Indeed it replaces an infinite set of equations by only one. Moreover, it is global, and it better exhibits the duality between electric potentials and ground state densities. This reduces to $\ro_{\p_1} = \ro_{\p_2}$ when one wants it to be independent of external potentials.

\subsection{Remark on the original proof}\label{rk}

The original proof of Hohenberg and Kohn is by contradiction, it assumes the non-degeneracy of the systems and uses strict inequalities in \eqref{e1} and \eqref{e2}. The authors precised that two different potentials could not lead to the same ground state, and Lieb \cite{Lieb83b} remarked that this fact relies on a unique continuation property, later proved in \cite{Garrigue19} for Coulomb systems. The proof using non-strict inequalities, first appearing in \cite{PinBokLud07}, is more general because it avoids to assume non-degeneracy of systems. Moreover, it is direct (proofs by contradiction are avoided when possible) and does not involve any additional argument compared to the original proof.

In the literature, there are many articles presenting Hohenberg-Kohn results but in which authors do not show that two different operators cannot produce the same ground states. This leads to incomplete proofs at best, but to false statements at worse. For instance the Hohenberg-Kohn theorems for interactions \cite{SieMul81}, and at positive temperature in the canonical case, hold by using Lemmas \ref{zeros} and \ref{lemfree2}, but in the case of non-local potentials, a counterexample can be found as presented in section \ref{counter}. The proof of a Hohenberg-Kohn result always follows the same scheme, and the only extra properties to show concern this step. They usually consist in a lemma, as Lemmas \ref{zeros}, \ref{zero2}, \ref{lemfree2}, \ref{lemfree}, see also the proof of Theorem \ref{spinhk}. Providing those results for some of the main Hohenberg-Kohn statements is one of the goals of this work.

\subsection{A semi-metric on the space of binding potentials}

In the following, we take $\Omega = \R^d$ for simplicity, and we show that the quantity involved in \eqref{assump} has special properties. We define the equivalence relation $\sim$ on the space $L^1\loc$ of functions and write $v \sim u$ if there is some constant $c$ such that $v = u + c$. We define the space of potentials
\begin{multline*}
\cV^N \df \Bigl\{ v \in L^p\loc(\R^d) \bigst v_- \in (L^p+L^{\ii})(\R^d), \\
H^N(v) \tx{ has a non degenerate fermionic ground state} \Bigr\} \big/\sim, 
\end{multline*}
where we identified potentials modulo constants, and where $p > \max(2d/3,2)$. We define the map
\begin{equation*}\rohk : 
\begin{array}{rcl}
\cV^N & \lra & (L^1 \cap L^q)(\R^d) \\
v & \longmapsto & \rohk(v),
\end{array}
\end{equation*}
which associates the unique ground state density to any potential $v \in \cV^N$, and where $q \df d/(d-2)$ if $d \ge 3$, $q$ can take any value in $[1,+\ii[$ if $d= 2$, and $q = +\ii$ if $d=1$.
The Hohenberg-Kohn theorem implies the injectivity 
of $\rohk$, therefore it is bijective on its image. This is the reason why it is commonly said that, at equilibrium, the ground state density contains the information of the external classical electric field. Now we define the function
\begin{align*}
\dis(v,u) \df - \int_{\R^d} (v-u) \bpa{\rohk(v)-\rohk(u)}
\end{align*}
if $\int u_+ \rohk(v)$ and $\int v_+ \rohk(u)$ are finite, and $\dis(v,u) \df + \ii$ otherwise. We remark that $\dis(v,u) \ge 0$ for any $v,u \in \cV^N$, using the inequalities \eqref{e1} and \eqref{e2} presented in the proof of the Hohenberg-Kohn theorem, and we remark that $\dis(v,u) = 0$ if and only if $v \sim u$. Thus $\dis$ is a semi-metric. For example, $\dis(v,u)$ is finite when applied to $v(x) = x^2$ and $u(x) = -Z_N \ab{x}\iv$ (for $Z_N$ large enough so that $u \in \cV^N$), because of the exponential decay of ground states when the potentials grow polynomially and in presence of a gap between the ground energy and the essential spectrum \cite{Simon82,Agmon,Hislop00}. Those are examples of standard physical potentials, but their $L^{p}+L^{\ii}$ distance is infinite. Therefore $\dis$ seems to be a natural \apo{physical distance}, because it enables to compare very different, but still physical, potentials.

In Theorem \ref{clHK} and in the definition of $\dis$, we can avoid the assumptions $\int (v_1)_+ \ro_{\p_2}$ and $\int (v_2)_+ \ro_{\p_1}$ to be finite, replacing them by the assumptions that $(v_1)_+$ and $(v_2)_+$ are at most polynomially increasing when $\ab{x} \ra +\ii$. Then by the exponential decay of ground states, $\int (v_1)_+ \ro_{\p_2}$ and $\int (v_2)_+ \ro_{\p_1}$ are automatically finite if $E_i \sle \min \sigma\ind{ess}\pa{H^N(v_i)}$.

\subsection{Hohenberg-Kohn for interactions} 

We consider $N$ particles in $\er{d}$, submitted to a two-body multiplication operator $\vde$, accounting both for interactions and external potentials. The corresponding $N$-body Hamiltonian takes the form 
\begin{equation}\label{ham}
H^N(\vde) = \sum_{i=1}^N -\Delta_{i} +  \sum_{1 \leq i \sle j \leq N} \vde(x_i,x_j),
\end{equation}
acting on $L^2(\er{dN})$. We will also use
\begin{equation}\label{ham}
H^N(v,w) = \sum_{i=1}^N -\Delta_{i} + \sum_{i=1}^{N} v(x_i) + \sum_{1 \leq i \sle j \leq N} w(x_i-x_j),
\end{equation}
on $L^2(\er{dN})$, which corresponds to $\vde(x,y) = \f{1}{N-1} \pa{v(x) + v(y)} + w(x-y)$. The two-body reduced density of a state $\p$, or \emph{pair density}, is defined by
\begin{align*}
\rod_{\p}(x,y) \df \f{N(N-1)}{2} \int_{\er{d(N-2)}} \abs{\p}^2(x,y,x_3,\dots,x_N) \d x_3 \cdots \d x_N.
\end{align*}
We can deduce the density from $\rod_{\p}$ by taking its marginal
\begin{align*}
\ro_{\p}(x) = \f{2}{N-1} \int_{\er{d}} \rod_{\p}(x,y) \d y.
\end{align*}
The energy of a state is coupled to the two-body potential $W$ only via $\rod_{\p}$, because 
\begin{align*}
	\biggps{\p, \sum_{1 \leq i \sle j \leq N} \vde(x_i,x_j) \p} = \int_{\R^{2d}} \vde\rod_{\p} .
\end{align*}

The next Hohenberg-Kohn theorem establishes a duality between $\vde$ and $\ro\de$, and between $(v,w)$ and $\ro\de$. A similar statement is present in \cite{SieMul81}, but with an incomplete proof, which needs Lemma \ref{zeros} $ii)$ (see the remark in section \ref{rk}). The knowledge of $\ro\de$ is sufficient to determine $\vde$ alone or the pair $(v,w)$, hence the ground state pair correlations $\ro\de$ contain all the information of the system.
\begin{theorem}[Hohenberg-Kohn for interactions]\label{hkdeux}\tx{ }

$i$) Let $q > \max(4d/3,2)$, and let $\vde_1, \vde_2 \in (L^q+L^{\ii})(\R^{2d})$ be even two-body potentials such that $H^N(\vde_1)$ and $H^N(\vde_2)$ have ground states $\p_1$ and $\p_2$. If
\begin{align*}
	\int_{\er{2d}}\bpa{\vde_1-\vde_2} \bpa{\ro_{\p_1}\de-\ro_{\p_2}\de} =0,
\end{align*}
then $$\vde_1 = \vde_2 +\f{2(E_1-E_2)}{N(N-1)}.$$

$ii$) Let $p > \max(2d/3,2)$ and let the potentials $v_1,v_2,w_1,w_2 \in (L^{p}+L^{\ii})(\er{d})$, $w_1,w_2$ even, be such that $H^N(v_1,w_1)$ and $H^N(v_2,w_2)$ have ground states $\p_1$ and $\p_2$. If
\begin{multline*}
	\int_{\R^d} (v_1-v_2)\pa{\ro_{\p_1}-\ro_{\p_2}} + \int_{\er{2d}}\pa{w_1-w_2}(x-y) \bpa{\ro_{\p_1}\de-\ro_{\p_2}\de}(x,y) \d x \d y  =0,
\end{multline*}
then there exists a constant $c \in \R$ such that
\begin{equation*}
\left\{
\begin{aligned}
& w_1 = w_2 + c, \\
& v_1 = v_2 + \f{E_1-E_2}{N} - \f{c(N-1)}{2}.
\end{aligned}
\right.
\end{equation*}
\end{theorem}

We stated the theorem in the whole space $\R^d$ for the sake of simplicity, but as Theorem \ref{clHK}, it holds for any open connected domain $\Omega \subset \R^d$, with any boundary condition. More precisely, knowing $\ro\de$ on $\Omega^2$ enables to know $v$ on $\Omega$ and $w$ on $\acs{x-y \st x,y \in \Omega}$. 

Pair DFT (or 2RDFT for two-body reduced density functional theory) was founded in \cite{SieMul81,Ziesche94,Ziesche96}, and further explored in \cite{GonSchVan96,LevZie01,Nagy02,Nagy03,NagAmo04,Furche04,AyeLev05,HigHig07,AyeNag07,AyeFue09,CheFri15} among other works. In particular, Mazziotti \cite{Mazziotti02,Mazziotti05,Mazziotti06,Mazziotti07,Mazziotti16} studied it extensively. This framework gives to the ground state two-body reduced density the central role. Theorem \ref{hkdeux} shows that this theory is well-posed.

\begin{proof}\tx{ }

$i)$ By the standard proof of the Hohenberg-Kohn theorem recalled above, we have $\bigps{\p_2,H^N(\vde_1) \p_2} = E_1$, so $\p_2$ is a ground state for $H^N(\vde_1)$, and $H^N(\vde_1) \p_2 = E_1 \p_2$. Taking the difference with ${H^N(\vde_2) \p_2 = E_2 \p_2}$ yields
\begin{equation*}
\pa{E_1-E_2+ \sum_{1 \leq i \sle j \leq N}^N \bpa{\vde_1-\vde_2}(x_i,x_j)}\p_2 = 0.
\end{equation*}
We use the unique continuation result \cite[Theorem 1.1]{Garrigue18}, in which we need an assumption of the type $\ab{W} \le \ep(-\Delta)^{\f{3}{2}-\delta}+c$ in $\R^{2d}$. We can apply \cite[Corollary 1.2]{Garrigue18} without $w$, replacing $v$ by $W$, and $d$ by $d'=2d$, hence the number $2d'/3$ = $4d/3$. The normalized function $\p_2$ thus cannot vanish on a set of positive measure and we have
\begin{equation*}
E_1-E_2+\sum_{1 \leq i \sle j \leq N}^N \bpa{\vde_1-\vde_2}(x_i,x_j)  = 0
\end{equation*}
	a.e. in $\R^{dN}$. We conclude by the following lemma, which also enables to prove $ii)$.
\end{proof}
\begin{lemma}\label{zeros}\tx{ }

$i)$ Let $\vde \in L^1\loc(\R^{2d})$ be even and such that
\begin{align}\label{eqlem2}
\sum_{1 \le i \sle j \le N} \vde(x_i,x_j) = 0,
\end{align}
a.e. in $\R^{dN}$. Then $\vde = 0$ a.e. in $\R^{2d}$.

$ii)$ If $v, w \in L^1\loc(\R^d)$, with $w$ even, and if
\begin{align}\label{eqlem}
\sum_{i=1}^N v(x_i) + \sum_{1 \le i \sle j \le N} w(x_i-x_j) = 0,
\end{align}
a.e. in $\R^{dN}$, then $v$ and $w$ are a.e. constant and verify
\begin{align*}
v + \f{N-1}{2} w = 0.
\end{align*}
\end{lemma}
\begin{proof}[Proof of Lemma \ref{zeros}]
$i)$ We consider $\vp \in \cC^{\ii}\ind{c}(\R^d)$ such that $\int \vp \neq 0$ and integrate \eqref{eqlem2} against $\vp^{\otimes N}$, which yields
\begin{align*}
\int \vde \vp^{\otimes 2} = 0.
 \end{align*}
We use it with $\vp = \chi + \ep \eta$, $\chi,\eta \in \cC^{\ii}\ind{c}(\R^d)$, $\int \chi \neq 0$ and $\ep$ small, and viewing the result as a polynomial in $\ep$, the coefficient in $\ep$ has to vanish, therefore
\begin{align}\label{mop2}
\int \vde \pa{\chi \otimes \eta}= 0.
 \end{align}
Let $f \in \cC^{\ii}\ind{c}(\R^d,\R_+)$ be a regularizing function such that $\int f = 1$. For $x_0 \in \R^d$ and $\ep > 0$, we denote by $f_{x_0}^{\ep}(x) \df \ep^{-d} f((x-x_0)/\ep)$ the translated and scaled function, and write $f^{\ep} \df f^{\ep}_0$. We apply \eqref{mop2} to $\chi = f_{x_0}^\ep$ and $\eta = f_{y_0}^{\ep}$ for any $x_0, y_0 \in \R^d$, to obtain
\begin{align*}
0 = \pa{\vde* \pa{ f^{\ep} \otimes f^{\ep}}}(x_0,y_0).
\end{align*}
We let $\ep \ra 0$, so $\vde * \pa{ f^{\ep} \otimes f^{\ep}} \lra \vde$ a.e. in $\R^{2d}$, and we get $\vde = 0$.

$ii)$ We apply $i)$ to $W(x,y) = \pa{N-1}\iv\pa{ v(x)+ v(y)} +  w(x-y)$, yielding 
\begin{align}\label{eqm}
v(x)+ v(y) + (N-1) w(x-y)=0
\end{align}
a.e. in $\R^{2d}$. This can be rewritten 
\begin{align*}
v(x+y) + v(y) = -(N-1) w(x) = -(N-1)w(-x) = v(-x+y) + v(y),
\end{align*}
and we have thus $v(x+y)=v(y-x)$. This implies $v(y) = v(y-2x)$ for a.e. $(x,y) \in \R^{2d}$, hence $v$ is constant and $w$ as well.
\end{proof}
A first consequence is that in a Kohn-Sham configuration, pair densities are different.
\begin{corollary}[Pair densities in the Kohn-Sham setting]
Let potentials $v_1,v_2,w \in (L^{p}+L^{\ii})(\er{d})$, with $p > \max(2d/3,2)$, where $w$ is even and not constant, such that $H^N(v_1,w)$ and $H^N(v_2,0)$ have ground states $\p_1$ and $\p_2$. Then $\rod_{\p_1} \neq \rod_{\p_2}$.
\end{corollary}

The proof can be done by contraposition, applying Theorem \ref{hkdeux} $ii)$. We note that we did not need to assume that $\ro_{\p_1} = \ro_{\p_2}$, as in the Kohn-Sham setting. This corollary is to compare with a similar result in Matrix DFT, claiming that it is not possible to reproduce the ground state one-body density matrix of an interacting system by a non-interacting system driven by a different electric potential. Indeed, for interacting Coulomb systems, Friesecke proved that the ground state one-body density matrix has infinite rank \cite[Theorem 2.1]{Friesecke03b}, whereas it is a finite rank projector for non-interacting systems.

When several types of particles are present, we can use the same principle. This illustrates the robustness of Hohenberg-Kohn results for interactions. For instance let us consider a mixture of two types of particles, $N \in \N \backslash \acs{0}$ of the first type and $M \in \N \backslash \acs{0}$ of the second, either fermions or bosons, represented by a wavefunction $\p \in L^2\pa{ (\R^d)^N \times (\R^d)^M, \C}$ (anti)symmetric in the first $N$ variables and (anti)symmetric in the last $M$ variables. The interactions between the particles of the first (resp. second) group are mediated by a function $w\ind{a}$ (resp. $w\ind{b}$), and the interactions between the two types are mediated by a third one $w\ind{ab}$. An external potential $v\ind{a}$ (resp. $v\ind{b}$) acts on the first (resp. second) type of particles. The difference in masses is implemented in a constant $\alpha \neq 0$. The Hamiltonian is
\begin{multline*}
H(v\ind{a},v\ind{b},w\ind{a},w\ind{b},w\ind{ab}) \df  \sum_{i=1}^N (-\Delta_i + v\ind{a}(x_i)) + \sum_{k=N+1}^{N+M} ( -\alpha \Delta_k + v\ind{b}(x_k)) \\
	+ \sum_{1 \le i \sle j \le N} w\ind{a}(x_i-x_j) + \sum_{N+1 \le k \sle l \le N+M} w\ind{b}(x_k-x_l) + \sum_{\substack{1 \le i \le N \\ N+1 \le k \le N+M}}  w\ind{ab}(x_i-x_k).
\end{multline*}
We define the pair function of the first type
\begin{align*}
	\ro^{(2)}_{\tx{a},\p}(x,y) \df \parmi{2}{N} \int_{\R^{d(N+M-2)}} \ab{\p}^2(x,y,x_3,\dots) \d x_3 \dots \d x_{N+M},
\end{align*}
the one of the second type
\begin{multline*}
	\ro^{(2)}_{\tx{b},\p}(x,y) \df \parmi{2}{M} \int_{\R^{d(N+M-2)}} \ab{\p}^2(x_1,\dots,x_N,x,y,x_{N+3},\dots) \\
	\times \d x_1 \dots \d x_N \d x_{N+3} \dots \d x_{N+M},
\end{multline*}
 and the pair function between the two types
\begin{multline*}
	\ro^{(2)}_{\tx{ab},\p}(x,y) \df NM \int_{\R^{d(N+M-2)}} \ab{\p}^2(x,x_2,\dots,x_N,y,x_{N+2},\dots) \\
	\times \d x_2 \dots \d x_N \d x_{N+2} \dots \d x_{N+M}.
\end{multline*}
Using the symmetries, the energy is thus
\begin{multline*}
	\ps{\p,H(v\ind{a},v\ind{b},w\ind{a},w\ind{b},w\ind{ab})\p} = \sum_{i=1}^N \int_{\R^d} \ab{\na_i \p}^2 + \alpha \sum_{k=N+1}^{N+M} \int_{\R^d} \ab{\na_k \p}^2 \\
	+ \f{2}{N-1} \int_{\R^d} \ro^{(2)}_{\tx{a},\p}(x,x)v\ind{a}(x)\d x+ \f{2}{M-1}\int_{\R^d}\ro^{(2)}_{\tx{b},\p}(x,x)v\ind{b}(x)\d x\\
	 + \int_{\R^{2d}} \ro^{(2)}_{\tx{a},\p}(x,y)w\ind{a}(x-y) \d x \d y + \int_{\R^{2d}} \ro^{(2)}_{\tx{b},\p}(x,y)w\ind{b}(x-y)\d x \d y  \\
	  + \int_{\R^{2d}} \ro^{(2)}_{\tx{ab},\p}(x,y)w\ind{ab}(x-y) \d x \d y.
\end{multline*}

\begin{theorem}[Hohenberg-Kohn for different particles]\label{types}
Let $p > \max(2d/3,2)$ and let the potentials $v_{\eta,i},w_{\eta,i},w_{\tx{ab},i} \in (L^{p}+L^{\ii})(\er{d})$ for $\eta \in \acs{\tx{a,b}}$ and  $i \in \acs{1,2}$, with $w_{\eta,i}, w_{\tx{ab},i}$ even, such that $H(v_{\tx{a},i},v_{\tx{a},i},w_{\tx{a},i},w_{\tx{b},i},w_{\tx{ab},i})$ have ground states $\p_i$. If 
\begin{align*}
\ro^{(2)}_{\tx{a},\p_1} = \ro^{(2)}_{\tx{a},\p_2}, \bhs \ro^{(2)}_{\tx{b},\p_1} = \ro^{(2)}_{\tx{b},\p_2},\bhs \ro^{(2)}_{\tx{ab},\p_1} = \ro^{(2)}_{\tx{ab},\p_2}, 
\end{align*}
then $v_{\eta,1}-v_{\eta,2}$, $w_{\eta,1}-w_{\eta,2}$, and $w_{\tx{ab},1}-w_{\tx{ab},2}$ are constant, for $\eta \in \acs{\tx{a,b}}$, and satisfy
\begin{multline*}
\f{N(N-1)}{2} (w_{\tx{a},1}-w_{\tx{a},2}) + \f{M(M-1)}{2} (w_{\tx{b},1}-w_{\tx{b},2})+ NM(w_{\tx{ab},1}-w_{\tx{ab},2})\\
+N(v_{\tx{a},1}-v_{\tx{a},2}) + M(v_{\tx{b},1}-v_{\tx{b},2})  = E_1 - E_2.
\end{multline*}
\end{theorem}

A slightly modified unique continution property is necessary for the proof, which takes into account the presence of $\alpha$. Also, the proof follows the same steps as for the standard Hohenberg-Kohn theorem, and requires the following lemma.
\begin{lemma}\label{zero2}
If $v\ind{a},v\ind{b}, w\ind{a},w\ind{b},w\ind{ab} \in L^1\loc(\R^d)$, with $w\ind{a},w\ind{b}$ and $w\ind{ab}$ even, and if
\begin{multline}\label{hypot}
\sum_{i=1}^N v\ind{a}(x_i) + \sum_{k=N+1}^{N+M} v\ind{b}(x_k)+ \sum_{1 \le i \sle j \le N} w\ind{a}(x_i-x_j) \\
+ \sum_{N+1 \le k \sle l \le N+M} w\ind{b}(x_k-x_l) + \sum_{\substack{1 \le i \le N  \\ N+1 \le k \le N+M}} w\ind{ab}(x_i-x_k)= 0,
\end{multline}
a.e. in $\R^{dN}$, then $v\ind{a},v\ind{b},w\ind{a},w\ind{b}$ and $w\ind{ab}$ are a.e. constant and verify
\begin{align*}
N v\ind{a} + Mv\ind{b} + \f{N(N-1)}{2} w\ind{a} + \f{M(M-1)}{2} w\ind{b}+ NM w\ind{ab} = 0.
\end{align*}
\end{lemma}
\begin{proof}
	We start by taking $\vp \in \cC^{\ii}\ind{c}(\R^d)$ such that $\int \vp = 1$, we multiply \eqref{hypot} by $\prod_{k=N+1}^{N+M} \vp(x_k)$ and we integrate over $(x_k)_{N+1 \le k \le N+M} \in \pa{\R^d}^M$, which yields
\begin{align*}
\sum_{i=1}^N \pa{ v\ind{a}(x_i) + M w\ind{ab} * \vp} + \sum_{1 \le i \sle j \le N} w\ind{a}(x_i-x_j) = c,
\end{align*}
where $c$ is a constant. Applying Lemma \ref{zeros}, $ii$), we know that $w\ind{a}$ is constant, and by symmetry, $w\ind{b}$ as well. We thus want to show that if
\begin{align*}
\sum_{i=1}^N v\ind{a}(x_i) + \sum_{k=N+1}^{N+M} v\ind{b}(x_k) + \sum_{\substack{1 \le i \le N  \\ N+1 \le k \le N+M}} w\ind{ab}(x_i-x_k)= 0,
\end{align*}
then $v\ind{a},v\ind{b}$ and $w\ind{ab}$ are constant. With similar notations as in the proof of Lemma \ref{zeros} $i)$, and choosing $f$ even, we take any $x,y \in \R^d$ and multiply the equation by 
 \begin{align*}
f^{\ep}_{x}(x_1) \dots f^{\ep}_{x}(x_N)f^{\ep}_{y}(x_{N+1}) \dots f^{\ep}_{y}(x_{N+M}),
 \end{align*}
 and then integrate over all coordinates, giving
 \begin{align*}
N v\ind{a} * f^{\ep}(x) + M v\ind{b} * f^{\ep}(y) + NM \pa{w\ind{ab} * f^{\ep} * f^{\ep}}(x-y) = 0.
 \end{align*}
 Letting $\ep \ra 0$ yields
\begin{align}\label{eqd}
 N v\ind{a} (x) + M v\ind{b} (y) + NMw\ind{ab}(x-y) = 0,
 \end{align}
for a.e. $(x,y) \in \R^d$. Using the fact that $w\ind{ab}$ is even, we have
\begin{align*}
N v\ind{a} (x) + M v\ind{b} (y) & = - NMw\ind{ab}(x-y) = -NMw\ind{ab}(y-x) \\
& = Nv\ind{a} (y) + M v\ind{b} (x),
 \end{align*}
hence
\begin{align*}
\pa{N v\ind{a} - Mv\ind{b}} (x)  = \pa{N v\ind{a} -M v\ind{b}} (y) 
 \end{align*}
a.e. so $Nv\ind{a}-Mv\ind{b} = c$ is constant. Equation \eqref{eqd} can then be rewritten
 \begin{align*}
 v\ind{b} (x) + v\ind{b}(y) + N w\ind{ab}(x-y) + c/M= 0,
 \end{align*}
which is similar to \eqref{eqm}. With the same argument as in the proof of Lemma \ref{zeros} $ii)$, we can proove that $v\ind{b}, v\ind{a}$ and $w\ind{ab}$ are constant.
\end{proof}

Many quantum models are replaced by approximate effective ones in which interactions are changed, because they are easier to study and are exact in some limits. Theorem \ref{types} shows that they cannot have the same ground state pair correlations, hence it provides a limit to their predictive power.

By measuring pair functions of bound ground states, one can thus exhaustively reconstruct the interactions between particles in non relativistic quantum settings, using an inverse procedure. In particular, this could be applied to the phonon-mediated effective interactions between electrons in a superconducting medium.

\section{Absence of Hohenberg-Kohn in Spin and Matrix DFT} 

We have seen that for the standard model, a Hohenberg-Kohn theorem holds. However, for other models, this is not necessarily the case. In Spin DFT, it is well-known that it does not hold \cite{CapVig01}, however we will show a partial Hohenberg-Kohn result in this model. We also give a counterexample to a Hohenberg-Kohn theorem in Matrix DFT.

\subsection{Partial Hohenberg-Kohn in Spin DFT} 

Spin DFT was founded by von Barth and Hedin in \cite{BarHed72}. It is a version of Density Functional Theory based on a variant of the Pauli Hamiltonian, in which the coupling between the current and the magnetic field is neglected. The only magnetic feature taken into account is the Zeeman interaction. It is a very active field of research in quantum physics and chemistry \cite{BarHed72,PanRaj72,RajCal73,GunLun76,Rajagopal80,EscPic01,KohSavUll05,PanSah12,PanSah15,ReiBorTel17}. This framework enables to study the relations between the ground state magnetization and the electromagnetic field.

We consider the Hamiltonian of Spin DFT
\begin{align*}
H^N (v,B) = \sum_{i=1}^N \bpa{- \Delta_i + \sigma_i \cdot B(x_i)+ v(x_i)} + \sum_{1 \leq i < j \leq N} w(x_i-x_j).
\end{align*}
 where $\cdot$ is the scalar product of $\R^d$. We fix the dimension $d = 3$ and recall the definition of the Pauli matrices 
\begin{align*}
\sigma^x = \begin{pmatrix} 0 & 1 \\ 1 & 0 \end{pmatrix}, \qquad \sigma^y = \begin{pmatrix} 0 & -i \\ i & 0 \end{pmatrix}, \qquad \sigma^z = \begin{pmatrix} 1 & 0 \\ 0 & -1 \end{pmatrix},
\end{align*}
 acting on one-particle two-components wavefunctions $\phi = \pa{\phi^{\upa}, \phi^{\doa}}^{T}$, where $\phi^{\upa}, \phi^{\doa} \in L^2(\er{d},\C)$ and $\int_{\R^d} \abs{\phi}^2 = 1$. In this case, $\sigma_i \cdot B(x_i) = B_{x}(x_i) \sigma^x_i +B_{y}(x_i) \sigma^y_i + B_{z}(x_i) \sigma^z_i$. The state of the system is described by antisymmetric and normalized wavefunctions $\p \in L^2_{\tx{a}}\pa{(\R^d \times \acs{\upa,\doa})^N}$. We introduce their one-body densities
 \begin{multline*}
 \ro^{\alpha \beta}_{\p}(x) \df \hspace{-0.5cm} \sum_{ (s_2,\dots,s_N) \in \acs{\upa,\doa}^{N-1}} \sum_{i=1}^N \int_{\er{d(N-1)}} \p\pa{\alpha,x;s_2,x_2;\dots} \ov{\p}\pa{\beta,x;s_2,x_2;\dots}\\
 \times \d x_2 \cdots \d x_N,
 \end{multline*}
 where $\alpha, \beta \in \acs{\upa, \doa}$, but for the sake of simplicity we will not always write the subscript $\p$. We remark that $\ro^{\upa\doa} = \ov{\ro^{\doa\upa}} \eqdef \xi$.
We define the density $\ro \df \ro^{\upa} + \ro^{\doa}$ and the magnetization
 \begin{align}\label{magne}
 m \df \begin{pmatrix} \ro^{\upa \doa} + \ro^{\doa\upa} \\ -i\pa{\ro^{\upa \doa} - \ro^{\doa\upa}} \\ \ro^{\upa \upa} - \ro^{\doa\doa} \end{pmatrix} = \mat{ 2\re \xi \\ 2 \im \xi \\\ro^{\upa \upa} - \ro^{\doa\doa} }.
\end{align}
 The quantum wavefunction is coupled to the external magnetic field only via the magnetization. Indeed, using fermionic statistics, we have
\begin{align*}
\biggps{\p,\sum_{i=1}^N \pa{\sigma_i \cdot  B(x_i)} \p} = \int_{\R^3} B \cdot m_{\p} .
\end{align*}

It is well-known that there is no complete Hohenberg-Kohn theorem in this model, due to a counterexample of Capelle and Vignale \cite{CapVig01} recalled below (see also \cite{EscPic01}). More explicitely, if the ground state densities and magnetizations of two systems are equal, this does not imply that the external potentials and magnetic fields are equal.
However, in the following theorem, we show that those assumptions imply a strong constraint on the external fields imposed on the system. This specifies the relation between $(v,B)$ and $(\ro,m)$.

\begin{theorem}[Partial Hohenberg-Kohn for Spin DFT]\label{spinhk}
Let $p >2$, let $w,v_1,v_2 \in (L^p+L^{\ii})(\er{3},\R)$ be potentials and $B_1,B_2 \in (L^p+L^{\ii})(\er{3},\R^3)$ be magnetic fields. We assume that $H^N(v_1,B_1)$ and $H^N(v_2,B_2)$ have ground states, which we denote by $\p_1$ and $\p_2$.
If
\begin{align}\label{condspin1}
\int_{\R^3}(v_1-v_2) (\ro_{\p_1} - \ro_{\p_2}) + \int_{\R^3} (B_1-B_2) \cdot (m_{\p_1}-m_{\p_2})  = 0,
\end{align}
then
\begin{align}\label{constraint}
\abs{B_1-B_2} \chi = \f{E_1-E_2}{N} + v_2 - v_1,
\end{align}
where $\chi$ is a measurable function taking its values in the discrete and finite set $\{-1,-1+\tfrac{2}{N},{-1+\tfrac{4}{N}},$ $ \dots, 1-\tfrac{2}{N}, 1\}$.
\end{theorem}
In particular, the condition \eqref{condspin1} is satisfied when $(\ro_{\p_1},m_{\p_1}) = (\ro_{\p_2},m_{\p_2})$.
We remark that on connected subsets where $v_1,v_2,B_1,B_2$ are continuous and where $B_1 \neq B_2$ everywhere, $\chi$ is continuous hence constant. Also, when $N$ is odd, $\chi$ can never vanish, thus if we have $v_1 -E_1/N= v_2 - E_2/N$ and $m_{\p_1} = m_{\p_2}$, we can deduce that $B_1 = B_2$. At fixed $v$, if we formally define a function $f_{\tx{HK}}(B) \df (m_{\p},E)$ associating to a magnetic field the ground state magnetization and energy, Theorem \ref{spinhk} implies that $f_{\tx{HK}}$ is injective, so it is bijective on its image. This shows that knowing $m_{\p}$ and $E$ enables to know $B$. Consequently, for a fixed $v$ and $N$ odd, all physical quantities are functionals of the ground state pair $(m,E)$.

To build a counterexample, Capelle and Vignale \cite{CapVig01} start from a system which ground state is an eigenvalue of the $z$ momentum operator $\sum_{i=1}^N \sigma_i^z$, and where the ground energy of $H^N(v_1,0)$ is isolated from the rest of the spectrum. For instance one can choose to start from $B_1 = 0$, with a binding $v_1$. Then they perturb this initial system by adding $B_2 = b e_z$ where $b \in \R$ is so small that there is no energy levels crossing, and they keep the electric potential $v_2 = v_1$ unchanged. In this case, our equation \eqref{constraint} becomes ${N b \chi = E_1-E_2}$, where $\chi$ is as in the statement of the theorem.

\begin{proof}
Following the same steps as in the proof of the standard Hohenberg-Kohn theorem, we can show that $\ps{\p_2,H^N(v_1,B_1) \p_2} = E_1$. Thus $\p_2$ verifies Schr\"odinger's equation of the first operator $H^N(v_1,B_1) \p_2 = E_1 \p_2$. Taking the difference with $H^N(v_2,B_2) \p_2 = E_2 \p_2$ yields
\begin{align*}
\pa{ \sum_{i=1}^N \pa{\indic_{2^N \times 2^N}v(x_i)  + B(x_i)\cdot\sigma_i  } } \p_2 = 0,
\end{align*}
where $v \df v_1 - v_2 + (E_2-E_1)/N$ and $B \df B_1 - B_2$. By strong unique continuation for the Pauli operator \cite{Garrigue19}, $\p_2$ does not vanish on sets of positive measure, thus we get
\begin{align}\label{ddd}
0 \in \sigma \pa{ \sum_{i=1}^N \pa{\indic_{2^N \times 2^N}v(x_i) + B(x_i)\cdot\sigma_i } }
\end{align}
	a.e. in $\R^{dN}$. We work in a cube $C \df \seg{0,a}^d$ for a fixed $a \in \R_+$. The function $x \mapsto v(x) + B(x) \cdot \sigma_i$ is continuous on a set $L_n$ such that $\ab{C \backslash L_n} \le 1/n$, by Lusin's theorem. Hence, $\sum_{i=1}^N v(x_i) + B(x_i) \cdot \sigma_i$ is continuous on $\pa{L_n}^N$. We take the limit $n \ra +\ii$ and use that the map giving the eigenvalues of a matrix is continuous, to infer that
\begin{align*}
0 \in \sigma \pa{  N v(x)  + B(x) \cdot \sum_{i=1}^N \sigma_i }
\end{align*}
for a.e. $x \in C$. We extend this conclusion a.e. in $x \in \R^3$ by letting $a \ra +\ii$. Eventually, one deduces that for a.e. $x \in \R^3$, at least one of the eigenvalues of $N v(x)  + B(x) \cdot \sum_{i=1}^N \sigma_i$ vanishes. Those eigenvalues are
\begin{align*}
\acs{ N v(x) + j \ab{B(x)}  }_{ j \in \acs{-N,-N+2,\dots,N-2,N } },
\end{align*}
by next Lemma \ref{lemmas}, for which we give a proof in the appendix.
\begin{lemma}[Local diagonalization of the Zeeman interaction]\label{lemmas}\tx{ }

For $B \in L^1\loc(\R^3,\R^3)$, $\sum_{i=1}^N B(x_i)\cdot\sigma_i  $ is diagonalizable in the spin variables, with eigenvalues
\begin{align*}
\acs{ \sum_{i=1}^N (-1)^{s_i} \ab{B(x_i)} }_{ (s_i)_{1 \le i \le N} \in \acs{0,1}^N  },
\end{align*}
a.e. in $\R^{dN}$.
\end{lemma}
Hence, for each $x\in \R^3$ such that $B(x) \neq 0$, and such that $v(x)$ and $B(x)$ are finite, there is a real number $\chi(x)\in$ $\acs{-1,-1+2/N,\dots,1-2/N,1 }$ respecting $v(x) + \chi(x) \ab{B(x)} =0$. On the set
 \begin{align*}
	 \acs{x \in \R^3 \st B(x) \neq 0, \ab{v(x)} + \ab{B(x)} \sle +\ii},
 \end{align*}
$\chi$ is measurable, and on the complementary space we choose $\chi(x) = 1$ for instance. The function $\chi$ is thus measurable.
\end{proof}

\subsection{Counterexample in Matrix DFT}\label{counter} 

Founded by Gilbert \cite{Gilbert75}, Reduced Density Matrix Functional Theory (RDMFT), hereafter called Matrix DFT, is extensively used in quantum chemistry \cite{MerKat77,DonPar78,Levy79,Valone80,YanZhaAye00,MarLat08,ShaDewLat08,BloPruPot13,Jones15,BalCanGro15,SchKamBlo17}. This method is similar to standard DFT but the central internal quantity is now the one-particle reduced density matrix (1RDM) which is defined by
\begin{equation*}
\gamma_{\Gamma}(x,y) \df N \int_{\er{d(N-1)}} \Gamma(x,x_2,\dots,x_N ; y,x_2,\dots,x_N) \d x_2 \cdots \d x_N
\end{equation*}
for mixed states $\Gamma$, and reduces to
\begin{equation*}
\gamma_{\p}(x,y) = N \int_{\er{d(N-1)}} \p(x,x_2,\dots,x_N) \ov{\p}(y,x_2,\dots,x_N) \d x_2 \cdots \d x_N
\end{equation*}
in the case of pure states $\Gamma = \ketbra{\p}{\p}$. The kinetic energy part of the exchange-correlation functional is not approximated, since the kinetic energy is an exact functional of the 1RDM. The drawback is that it is computationally more expensive than standard DFT, because 1RDM kernels have two space arguments instead of one. The external non-local potentials at stake in this framework are one-body operators of
 \begin{align*}
\cG \df \acs{ G = G^* \st \cD\bpa{\ab{G}^{\ud} } \subset H^1(\R^d), \hs\hs \forall \ep >0 \hs\hs \exists c_{\ep} \ge 0 \hs\hs  \ab{ G} \le -\ep \Delta + c_{\ep} }.
 \end{align*}
The multiplication by $v \in (L^{d/2}+L^{\ii})(\R^d)$, or a gradient magnetic operator $-2iA \cdot \na$ where $A \in (L^d+L^{\ii})(\R^d)$, belong to this class. One obtains the full Pauli Hamiltonian when $G = -2 i A \cdot \na + A^2 + \sigma \cdot \rot A$. Also, we can choose it to act by multiplication on the momentum space $(\hat{G \p})(k) = g(k) \hat{\p}(k)$ with $\ab{g} \le \ep k^2 + c_{\ep}$ for any $\ep > 0$. An operator $G \in \cG$ acts on many-body functions by the second quantized $G^{(N)} \in \cL(L^2(\R^{dN}))$ of $G$ on the $N$-particle sector
\begin{equation}\label{nonloc}
G^{(N)} = \sum_{i=1}^N G_i,
\end{equation}
where each $G_i$ is a copy of $G$ acting on the $i^{\tx{th}}$ body, that is $G_i \df \indic \otimes \dots\otimes G \otimes \dots \otimes \indic$. For instance, when $G$ has a kernel, 
 \begin{align*}
\pa{G^{(N)} \p}(x_1,\dots,x_N) = \sum_{i=1}^N  \int_{\R^d} G(x_i,y) \p(x_1,\dots, x_{i-1}, y , x_{i+1}, \dots, x_N) \d y.
 \end{align*}
The Hamiltonian in this formalism is
\begin{equation}\label{ham}
H^N(G) = \sum_{i=1}^N -\Delta_{i} + \sum_{i=1}^N G_i + \sum_{1 \leq i \sle j \leq N} w(x_i-x_j),
 \end{equation}
 and the operators in $\cG$ are such that all $H^N(G)$ have the same form domain $H^1(\R^d)$. A state is coupled to the external non-local potential only via its one-body density matrix, and the interaction energy of $\p$ with $G$ is equal to $\tr G \gamma_{\p}$.

In his pioneering work, Gilbert \cite{Gilbert75} wrote that the proof of Hohenberg and Kohn cannot be extended to $G \mapsto \gamma$, but he remarked that the injectivity $\p \mapsto \gamma_{\p}$ holds, when $\p$ is taken in the set of ground states of $H^N(G)$'s, and where $G$'s are chosen in $\cG$ and produce non-degenerate ground states. Here we provide a counter-example showing that the map $G \mapsto \gamma$ is not injective. This implies that in its most general formulation, Matrix DFT is ill-posed because it does not respect uniqueness.

Let $G_1 \in \cG$ be such that $H(G_1)$ has a unique ground state $\p_1$ isolated from the rest of the spectrum. 
We assume that 
 \begin{align}\label{assumption}
 \exists \phi \in H^1(\R^d), \mediumint \ab{\phi}^2 = 1 \tx{ such that } \pa{\phi \wedge L_{\tx{a}}^2(\R^{d(N-1)})} \perp \p_1.
 \end{align}
The assumption \eqref{assumption} is equivalent to the existence of an orthonormal basis $\acs{\phi_i}_{i \in \N}$ such that $\phi_1$ never appears in the decomposition of $\p_1$ on the basis built from $\acs{\phi_i}_{i \in \N}$. We take $G_2 = G_1 +  \ep \proj{\phi}$. We have $\proj{\phi} \in \cG$ and
\begin{align*}
H^N(G_2) \p_1 = H^N(G_1) \p_1 +  \ep \pa{\sum_{i=1}^N \proj{\phi}_i } \p_1 = E_1\p_1.
\end{align*}
We also have $H^N(G_1) \le H^N(G_2)$ in the sense of forms, therefore by the min-max principle, $E_1 \le E_2$ and $\p_1$ is a ground state of $H^N(G_2)$. Hence we have two Hamiltonians having different non-local potentials but the same ground state. We get counterexamples if we can find $w$ and $G_1$ such that \eqref{assumption} holds. This is the case when $w=0$ because then, $\p_1$ is a Slater determinant. We have thus here a class of counterexamples to the Hohenberg-Kohn property $G \mapsto \gamma$. This argument also works for the corresponding theory considering mixed states, because when $w=0$, the ground state is still a Slater determinant. We note that in \cite[Theorem 2.1]{Friesecke03b}, Friesecke proved that for interacting Coulomb systems and denoting by $\p$ the ground state, $\rg \gamma_{\p} = +\ii$, but it is not proved that $\im \gamma_{\p}$ spans the full space, in which case the condition \eqref{assumption} would be wrong for Coulomb interactions. 
This uniqueness property for non-local potentials when $w=\ab{\cdot}\iv$ is thus an open problem, one could also try to determine it in the simpler Hartree-Fock model for instance.

\section{Recovering Hohenberg-Kohn in Thermal DFT}\label{thdft} 
In this section we show that Hohenberg-Kohn theorems which do not hold at zero temperature are true when $T >0$. Founded by Mermin \cite{Mermin65}, Thermal DFT is the Density Functional Theory framework dealing with a fixed positive temperature. This formalism was used in several works \cite{PitProFlo11,BalCanGro15,RuizSky13,PriPitGro14,BurSmiGra16,Dharma16,SmiSagBur18}, in particular to model warm dense matter \cite{GraDesMic14}, and nuclear matter.
We provide here a new formalism of Thermal DFT enabling to work with systems having possibly different temperatures.

The set of $N$-particle canonical states is 
\begin{align*}
\cS_{\tx{can}}^N \df \acs{ \Gamma \in \cL\bpa{L^2_{\tx{a}}(\R^{dN})} \bigst 0 \le \Gamma=\Gamma^*, \tr \Gamma =1},
\end{align*}
 and the set of grand canonical states is 
\begin{align*}
\cS_{\tx{gc}} \df \acs{ \Gamma \in \cL\pa{\C \oplus_{N=1}^{+\ii}L^2_{\tx{a}}(\R^{dN}) } \bigst 0 \le \Gamma=\Gamma^*, \tr \Gamma =1}.
\end{align*}
The entropy is denoted by $S_{\Gamma} \df - \tr \Gamma \ln \Gamma$. We denote the particle number operator by $\cN$, the average number of particles of a state $\Gamma$ is $\tr \cN \Gamma$. We define the $k$-particle matrix densities
 \begin{align*}
 \Gamma^{(k)} \df \sum_{k \ge n} \parmi{n}{k} \tr_{k+1 \ra n} \Gamma,
 \end{align*}
 and the density
\begin{align*}
\ro_{\Gamma}(x) \df \Gamma^{(1)}(x,x).
\end{align*}
In the grand canonical case, for states $\Gamma = G_0 \oplus G_1 \oplus \cdots$ where $G_n \geq 0$, we define the partial densities $\ro_n(x) \df n \pa{\tr_{2 \ra n} G_n} (x,x)$, and thus $\ro_{\Gamma}(x) = \sum_{n\geq 1} \ro_n(x)$. Also, the number of particles is the total mass of the density $\int \ro_{\Gamma}  = \tr \cN \Gamma$. We define the second quantized operators $\bbK \df \oplus_{n=1}^{+\ii} \sum_{i=1}^n -\Delta_i$, $\bbV \df \oplus_{n=1}^{+\ii} \sum_{i=1}^n v(x_i)$, $\bbW \df \oplus_{n \ge 2} \sum_{1 \le i \sle j \le n} w(x_i-x_j)$, and
\begin{align*}
\bbH(v) \df \bbK + \bbV + \bbW.
\end{align*}
The wavefunction is coupled to the external potential only via its density, $\tr \bbV \Gamma = \int v \ro_{\Gamma}$. We denote by $H^N(v)$ the usual $N$-body Hamiltonian \eqref{opcl}.
The free energy of canonical states is
\begin{align*}
\cE_{v,T}^N(\Gamma_N) & \df \tr H^N(v) \Gamma_N - T S_{\Gamma_N} \\
& = \tr H^N(0) \Gamma_N + \int v \ro_{\Gamma_N}  - T S_{\Gamma_N}.
\end{align*}
The free energy of grand canonical states is
\begin{align*}
\cE_{v,T}(\Gamma) & \df \tr \bbH (v) \Gamma - TS_{\Gamma} \\
& = \tr \bbH(0) \Gamma + \int v\ro_{\Gamma} - TS_{\Gamma}.
\end{align*}
We can add a chemical potential in the manner $\cE_{v,T,\mu}(\Gamma) \df \cE_{v,T}(\Gamma) - \mu \tr \cN \Gamma$ but $\cE_{v,T,\mu} = \cE_{v-\mu,T}$ so without loss of generality, we can set $\mu = 0$ (or change $v \ra v - \mu$ everywhere), and work with $\cE_{v,T}$.
By Gibbs' variational principle, when $$\tr e^{-\bbH(v)/T}, \tr e^{-H^N(v)/T} \sle +\ii,$$ those functionals have a unique minimizer, called the Gibbs state. It is $\Gamma = Z\iv e^{-\bbH(v)/T}$ in the grand canonical ensemble, and $\Gamma_N = Z\iv e^{-H^N(v)/T}$ in the canonical one, where $Z \df \tr e^{-\bbH(v)/T}$ (resp. $Z \df \tr e^{-H^N(v)/T}$) is the partition function.

\subsection{A Hohenberg-Kohn theorem for $(v,T) \mapsto (\ro,S)$} 

We can now show a complete Hohenberg-Kohn theorem for Thermal DFT, extending Mermin's statement \cite{Mermin65} to systems having different temperatures, and to the canonical case. We define the exponent
\begin{equation}\label{exp}
\left\{
\begin{array}{ll}
p = d/2 & \tx{ for } d \ge 3,\\
p > 1 & \tx{ for } d =2, \\
p = 1 & \tx{ for } d =1.
\end{array}
\right.
\end{equation}

\begin{theorem}[Hohenberg-Kohn at positive temperatures]\label{hkt}
Let $T_1,T_2 >0$, and $p$ as in \eqref{exp}. Let $v_1,v_2 \in L^p\loc(\R^d)$, $w,(v_1)_-,(v_2)_- \in (L^{p}+L^{\ii})(\R^d)$ be potentials such that $\tr e^{-\bbH(v_j)/T_j} < +\ii$ for $j \in \acs{1,2}$. We denote by $\Gamma_1, \Gamma_2$ the grand canonical Gibbs states corresponding respectively to $\cE_{v_1,T_1}$ and $\cE_{v_2,T_2}$, and we assume that $\int (v_1)_+ \ro_{\Gamma_2}$ and $\int (v_2)_+ \ro_{\Gamma_1}$ are finite. If 
\begin{align*}
-\pa{T_1-T_2} \pa{S_{\Gamma_1} - S_{\Gamma_2}}  +\int_{\R^d} (v_1-v_2) (\ro_{\Gamma_1}-\ro_{\Gamma_2})= 0,
\end{align*}
then $T_1 = T_2$, $v_1 = v_2$, $Z_1=Z_2$ and $\Gamma_1 = \Gamma_2$. In the canonical setting, with similar assumptions we deduce that $T_1 = T_2$, $v_1 = v_2+ T_1 \ln \f{Z_2}{Z_1}$ and $\Gamma_1 = \Gamma_2$.
\end{theorem}
This result shows a duality between the internal equilibrium state quantities $(\ro,S)$ and the external imposed quantities $(v,T)$. The equilibrium state density and entropy of a quantum system contains the information of $(v,T)$. In particular, all physical quantities are functionals of the ground state pair $(\ro,S)$. Nevertheless, we are not aware of any functional of the ground state energy defined in terms of $(\ro,S)$, approximating the exact one. 

\begin{proof} By the standard proof of the Hohenberg-Kohn theorem, we show that $\cE_{v_1,T_1}(\Gamma_2) = E_1$ using that $\Gamma_2$ belongs to the quadratic form domain of $\bbH(v_1)$. Gibbs minimizers are unique, this can be proved by using Klein's inequality \cite[(5.3)]{Ruelle} or Pinsker, so $\Gamma_1 = \Gamma_2$, that is
\begin{align*}
Z_1\iv e^{-(\bbT + \bbW + \bbV_1)/T_1} = Z_2\iv e^{-(\bbT + \bbW + \bbV_2)/T_2},
\end{align*}
on the antisymmetric Fock space. Restricting to the zero-body sector yields $Z_1 = Z_2$, so we can rewrite
\begin{align*}
e^{-\inv{T_1}(\bbT + \bbW +\bbV_1)} = e^{-\f{\alpha}{T_1} \pa{ \bbT + \bbW +\bbV_2}},
\end{align*}
with $\alpha \df T_1 / T_2$. By injectivity of the exponential map for operators, we obtain
\begin{align*}
\bbT + \bbW +\bbV_1 = \alpha\pa{\bbT + \bbW +\bbV_2},
\end{align*}
on the antisymmetric Fock space. Restricting to the one-body sector gives
\begin{align*}
-\Delta +  v_1  = \alpha\pa{ -\Delta +v_2},
\end{align*}
and
\begin{align*}
(\alpha-1) \Delta  =  \alpha v_2 - v_1.
\end{align*}
In $H^1\loc(\R^d)$, the left-hand side is translation invariant, thus so is the right hand side, and we deduce that $\alpha = 1$, and then $v_1 = v_2$.

If we consider the canonical setting instead of the grand canonical one, we only have a constraint on the antisymmetric $N$-body sector but the proof is the same. Indeed, the same reasoning leads to
\begin{multline}\label{eqa}
\sum_{i=1}^N  \bpa{-\Delta_i + v_1(x_i)} + \sum_{1 \le i \sle j \le N} w(x_i-x_j) \\
= \alpha \pa{\sum_{i=1}^N  \bpa{-\Delta_i + v_2(x_i)} + \sum_{1 \le i \sle j \le N} w(x_i-x_j) - N T_2 \ln \f{Z_1}{Z_2}}
\end{multline}
	in $L^2\ind{a}(\R^{dN})$. We conclude by using the following lemma.
\end{proof}
\begin{lemma}\label{lemfree2}
	Let $A,B \in L^1\loc(\R^3,\R^3)$, $v, w \in L^1\loc(\R^3)$, with $w$ even, and $\alpha \in \R$ be such that 
\begin{multline*}
\indic_{2 \times 2}  \pa{ \sum_{j=1}^N \bpa{\alpha (-\Delta_i) -2i  A(x_i) \cdot \na_i +  v(x_i)} + \sum_{1 \le j \sle k \le N} w(x_j-x_k)} \\
+ \sum_{j=1}^N \sigma_j \cdot B(x_j) = 0,
\end{multline*}
	as an operator of $L^2\ind{a}(\R^{dN})$ (or of the bosonic counterpart). Then $\alpha = 0$, $A = 0$, $B=0$, $v$ and $w$ are constant and verify $v + w(N-1)/2=0$.
\end{lemma}
\begin{proof}
	We only provide a proof in the fermionic case, since the bosonic one follows from similar arguments. We choose a basis of the spin variable $(s_1, \dots, s_N)_{s_i \in \acs{ \upa,\doa}}$ in which the first element is $\ket{\upa,\dots,\upa}$ and the second one is $\ket{\upa,\dots,\upa,\doa}$. We consider a purely spatial antisymmetric function $\psi \in L\ind{a}^2(\R^{dN})$ and form the fermionic function $\p \df \psi \otimes \sum_{s_i \in \acs{\upa,\doa}} \ket{s_1,\dots,s_N}$, symmetric in spin and antisymmetric in space. We apply \eqref{eqa}, seen as a matrix in the previous spin basis, to our function, and we substract the second component of the resulting vector from the first one, which yields
\begin{align*}
(B_z-i B_y)(x_N) \psi(x_1,\dots,x_N) = 0,
\end{align*}
a.e. in $\R^{dN}$. Since this holds for any antisymmetric $\psi$, we deduce that $B_z = i B_y$, and since $B$ is real, then $B_z = B_y = 0$. A similar argument leads to $B_x = 0$. We thus have
\begin{align}\label{lam}
  \sum_{j=1}^N \bpa{\alpha (-\Delta_i) -2i  A(x_i) \cdot \na_i +  v(x_i)} + \sum_{1 \le j \sle k \le N} w(x_j-x_k) = 0
 \end{align}
on $L^2\ind{a}(\R^{dN})$. We choose a smooth localizing function $\chi$ with support in the ball $B_r \subset \R^d$ for some $r \in \R_+$, we take centers $y_1,\dots,y_N \in \R^3$ separated by distances larger than $r$, we form the orbitals $\phi_i(x) \df (\chi e^{-a\ab{\cdot}})(x-y_i)$ and the $N$-body wavefunction $\Phi \df \wedge_{i=1}^N \phi_i$. We apply \eqref{lam} to it, yielding
\begin{align}\label{eqa1}
 V \Phi + \sum_{i=1}^N \phi_1 \wedge \cdots \wedge \pa{ -\alpha \Delta -2i A \cdot \na} \phi_i \wedge \cdots \wedge \phi_N  = 0,
 \end{align}
where $V \df \sum_{j=1}^N v(x_i) + \sum_{1 \le j \sle k \le N} w(x_j-x_k)$. We compute
 \begin{multline*}
	 \pa{ -\alpha \Delta -2i A \cdot \na} \pa{ \chi(x) e^{-a\ab{x}}}  = \f{e^{-a\ab{x}}}{\ab{x}} \big(   -a^2\alpha \ab{x} \chi  \\ 
	 + 2a \pa{ \alpha \pa{ \chi + x \cdot \na \chi} + i \chi A \cdot x } -\alpha \ab{x} \Delta \chi -2i\ab{x} A \cdot \na \chi  \big).
 \end{multline*}
	Finally we evaluate \eqref{eqa1} in the neighborhood of one of the $y_i$'s, and as a polynomial in $a$, the resulting equation implies $\alpha = 0$, and then $A = 0$ a.e. in an open region, which can be extended to the whole space $\R^d$ by moving the $y_i$'s. Hence we can also deduce that $V = 0$ a.e. in $\R^{dN}$, and we conclude by using Lemma \ref{zeros}.
\end{proof}

We could have stated the theorem with the assumption $T_1, T_2 \in \R_+$. In this case, we should also assume that $v_1,v_2,w \in L^{q}\loc(\R^d)$ with $q > \max(2d/3,2)$ in order to use the standard Hohenberg-Kohn theorem in the case $T_1 = T_2 = 0$. Indeed, if $T_1 = 0$ and $T_2 > 0$, the minimizer $\Gamma_1$ of the first functional is pure, but it minimizes $\cE_{v_2,T_2}$ and is in its variational minimization set. Consequently it is equal to the Gibbs state $\Gamma_2$ which is not pure. There is a contradiction and therefore this configuration is impossible, $T_1$ and $T_2$ are either both equal to zero, or both strictly positive. We emphasize the fact that when temperatures are strictly positive, the proof does not involve any unique continuation argument, and is thus much simpler than at zero temperature.

We showed that the knowledge of the ground state pair $(\ro,S)$ contains the information of $(v,T)$. We conjecture now that the knowledge of $\ro$ alone does not contain the information of $(v,T)$, i.e. that the map $(v,T) \mapsto \ro$, giving the density of the Gibbs state for the pair $(v,T)$, is not injective. 

\subsection{Lifting ill-posedness at positive temperature} 

We are going to see now that increasing the temperature generically removes the ill-posedness of Hohenberg-Kohn theorems. We begin to analyze Current DFT, for which there are counterexamples to a Hohenberg-Kohn theorem \cite{CapVig02,LaeBen14}. We define the paramagnetic current of a wavefunction $\p$ by
 \begin{align*}
	 j_{\p}(x) & \df \im \sum_{(s_k)_{1 \le k \le N} \in \acs{\upa,\doa}^N} \sum_{i=1}^N \int_{\R^{d(N-1)}} \ov{\p^{s_k}} \na_i \p^{s_k}  \d x_1 \cdots \d x_{i-1} \d x_{i+1} \cdots \d x_N,
 \end{align*}
 and its total current by $j_{\tx{t}} \df j + \ro A + \rot m$, where the magnetization of a state is defined in \eqref{magne}. The Hamiltonian that we consider is the many-body Pauli operator
\begin{align*}
	H^N (v,A,w)   \df \sum_{j=1}^N \pa{\bpa{\sigma_j \cdot \pa{-i\na_j + A(x_j)} }^2 +  v(x_j)}+ \sum_{1 \le i \sle j \le N} w(x_i-x_j),
\end{align*}
which is $\bbH(v,A,w) \df \oplus_{n=1}^{+\ii} \sum_{j=1}^n  \bpa{\sigma_j \cdot (-i\na_j + A(x_j))}^2 + \bbV + \bbW$ in the grand canonical setting.
In the case of canonical and grand canonical states, we can also define a current and a magnetization $j_{\Gamma},m_{\Gamma}$ by the decomposition of $\Gamma$ into pure states.
We choose the Coulomb gauge, and using fermionic statistics and the grand canonical ensemble, the kinetic energy of a state can be expressed by
 \begin{multline*}
\tr \pa{\sum_{j=1}^N \bpa{\sigma_j \cdot \pa{-i\na_j + A(x_j)} }^2 } \Gamma  \\
	 =  \tr (-\Delta) \Gamma + \int A^2 \ro_{\Gamma}  +  \int A \cdot (2j_{\Gamma} +\rot m_{\Gamma}).
 \end{multline*}

\begin{theorem}[Hohenberg-Kohn for interacting Pauli systems at $T >0$]\label{thm2}
Let $T_1,T_2>0$, let $v_1,v_2 \in L^{3/2}\loc(\R^3)$, let $A_1,A_2 \in (L^3+L^{\ii})(\R^3)$, $(v_1)_-,(v_2)_-$, $w_1,w_2,$ $\ab{\rot A_1}, \ab{\rot A_2} \in (L^{3/2}+L^{\ii})(\R^3)$ be potentials such that $w_1$ and $w_2$ are even and such that the two grand canonical partition functions $\tr e^{-\bbH(v_j,A_j,w_j)/T_j}$ are finite. We denote by $\Gamma_1, \Gamma_2$ the grand canonical Gibbs states corresponding to the free energies in the Pauli model with temperature, and we assume that all the quantities involved in \eqref{assumption1} are finite. If
\begin{multline}\label{assumption1}
- (T_1-T_2)(S_{\Gamma_1}-S_{\Gamma_2}) + \int_{\R^d} (v_1-A_1^2-v_2+A_2^2)(\ro_{\Gamma_1}-\ro_{\Gamma_2}) \\
+ \int_{\R^d} (A_1 - A_2)\cdot (2j_{\tx{t},\Gamma_1} - \rot m_{\Gamma_1} - 2j_{\tx{t},\Gamma_2} +\rot m_{\Gamma_2} ) \\
+ \int_{\er{2d}}\pa{w_1-w_2}(x-y) \bpa{\ro_{\Gamma_1}\de-\ro_{\Gamma_2}\de}(x,y) \d x \d y =0,
\end{multline}
then there exists $c\in \R$ such that $T_1 = T_2$, $A_1 = A_2$, $w_1 = w_2 + c$, $v_1 = v_2 -c(N-1)/2$, $Z_1=Z_2$ and $\Gamma_1 = \Gamma_2$. In the canonical setting, we deduce that $T_1 = T_2$, $A_1 = A_2$, $w_1 = w_2 + c$, $v_1 = v_2 + T_1 \ln \f{Z_2}{Z_1} - c(N-1)/2$, and $\Gamma_1 = \Gamma_2$.
\end{theorem}
The proof follows the one of Theorem \ref{hkt} and uses Lemma \ref{lemfree2}. The assumption \eqref{assumption1} is fulfilled when $$\bpa{\ro_{\Gamma_1},j_{\tx{t},\Gamma_1},m_{\Gamma_1},S_{\Gamma_1}} = \bpa{ \ro_{\Gamma_2},j_{\tx{t},\Gamma_2},m_{\Gamma_2},S_{\Gamma_2}}$$ for instance. This remedies to the ill-posedness of the problem at zero temperature. Following ideas of Ruggenthaler and Tellgren \cite{Ruggenthaler15,Tellgren18,Garrigue19}, another way to do so is by adding internal magnetic degrees of freedom.

The last theorem (without spin) holds for classical canonical and grand canonical systems, and the proof is the same. A similar statement with different interactions is present in Henderson \cite{Henderson74}. In the classical canonical case, states are symmetric probability measures ${\mu \in \cP\ind{s}(\R^{dN} \times \R^{dN})}$ on the phase space, and the current is 
 \begin{align*}
j_{\mu}(x) \df N \int_{\R^{d(N-1)} \times \R^{dN}} p_1 \mu(x,p_1,x_2,p_2,x_3,p_3,\dots) \d x_2 \dots \d x_N \d p_1 \dots \d p_N.
 \end{align*}
 Other internal variables can be defined similarly, and the assumption
\begin{multline*}
-(T_1-T_2)(S_{\mu_1}-S_{\mu_2}) + \int_{\R^{d}} (v_1+A_1^2-v_2-A_2^2)\pa{\ro_{\mu_1}-\ro_{\mu_2}}  \\
+ 2 \int_{\R^d} (A_1-A_2) \cdot (j_{\mu_1}-j_{\mu_2}) \\
 + \int_{\R^d \times \R^d}\pa{w_1-w_2}(x-y) \bpa{\ro_{\mu_1}\de-\ro_{\mu_2}\de}(x,y) \d x \d y  =0,
\end{multline*}
has the same consequences as in the quantum counterpart. In particular, it is ensured that $(\ro\de,j,S)$ contains all the information of a classical system at equilibrium. Without magnetic fields and at fixed temperatures, $\ro\de$ alone contains all the information. It is nevertheless important that temperatures are striclty positive in the classical case, whereas they can vanish in the quantum case.

A similar version of the next theorem at fixed temperature in the grand canonical case was presented in \cite{BalCanGro15}. It shows that at positive temperature, Matrix DFT is well-posed.

\begin{theorem}[Hohenberg-Kohn for non-local potentials at $T>0$]\label{thm1}
Let $T_1,T_2 >0$, let $p$ be as in \eqref{exp}, let $v \in L^{p}\loc(\R^d)$ with $v_- \in (L^{p}+L^{\ii})(\R^d)$ be a trapping potential, let $G_1,G_2 \in \cG$ be such that $G_j \ge \ep\Delta - c_{\ep}$ for any $\ep >0$, and such that $\tr e^{-\pa{\bbH(v) + \oplus_{n=0}^{+\ii} \sum_{i=1}^n \pa{G_j}_i }/T_j}$ are finite. We denote by $\Gamma_1, \Gamma_2$ the grand canonical Gibbs states corresponding to the free energies with temperature in this model, and we assume that $\int G_1 \gamma_{\Gamma_2}$ and $\int G_2 \gamma_{\Gamma_1}$ are finite. If
\begin{align}\label{assumptionMatrix}
- (T_1-T_2)(S_{\Gamma_1}-S_{\Gamma_2}) + \tr (G_1-G_2) (\gamma_{\Gamma_1}-\gamma_{\Gamma_2}) = 0,
\end{align}
then $T_1 = T_2$, $G_1 = G_2$, $Z_1=Z_2$ and $\Gamma_1 = \Gamma_2$. In the canonical setting, we deduce that $T_1 = T_2$, $G_1 = G_2 + T_1 \ln \f{Z_2}{Z_1}$ and $\Gamma_1 = \Gamma_2$.
\end{theorem}

The proof follows the one of Theorem \ref{hkt} and uses the following lemma.
\begin{lemma}\label{lemfree}
Let $G$ be a self-adjoint operator for which $G \ge \ep \Delta - c_{\ep}$ in the sense of forms in $L^2(B)$ for some ball $B \subset \R^d$, for any $\ep > 0$ and some $c_{\ep} \ge 0$. If 
\begin{align}\label{eqal}
	\sum_{i=1}^N \pa{-\alpha \Delta_i +  G_i} = 0
 \end{align}
	in $L\ind{a}^2(B^N)$ (or in the bosonic counterpart) for some $\alpha \in \R$, then $\alpha = 0$ and $G = 0$ on $B$.
\end{lemma}
\begin{proof}
We treat the fermionic case, since the bosonic one follows from similar arguments. Let $(\phi_i)_{i \in \N}$ be an orthonormal basis of $L^2(B)$. We apply \eqref{eqal} to $\wedge_{i=1}^N \phi_i$ and take the scalar product with this same vector, which gives 
	\begin{align}\label{mama}
  \sum_{i=1}^N \ps{ \phi_i, \pa{-\alpha \Delta + G} \phi_i} = 0.
 \end{align}
	The same procedure applied to $\wedge_{i=2}^{N+1} \phi_i$ yields $\sum_{i=2}^{N+1} \ps{ \phi_i, \pa{-\alpha \Delta + G} \phi_i} = 0$ and taking the difference between those equations, we get 
 \begin{align*}
		\ps{\phi_1, \pa{-\alpha \Delta + G} \phi_1} = \ps{\phi_{N+1}, \pa{-\alpha \Delta + G} \phi_{N+1}}.
 \end{align*}
 Similarly, we have $\ps{\phi_i, \pa{-\alpha \Delta + G} \phi_i} = \ps{\phi_j, \pa{-\alpha \Delta + G} \phi_j}$ for any $i,j \in \N$, and using \eqref{mama} again, we conclude that $\ps{\phi_i, \pa{-\alpha \Delta + G} \phi_i} = 0$ for any $i \in \N$. By polarization, we deduce that $-\alpha \Delta +G = 0$ and we reduced the problem to the $1$-particle case. With $\ep = \alpha/2$, we have $G = \alpha \Delta \ge \f{\alpha}{2} - c_{\alpha/2}$, and thus $\f{\alpha}{2}(-\Delta) \le c_{\alpha/2}$, which implies $\alpha = 0$ and $G = 0$.
\end{proof}

We refer to \cite{GieRug19} for a review about Matrix DFT at positive temperature.

\section*{Appendix} 

\subsubsection*{Proof of Lemma \ref{lemmas}}
We consider the canonical spin basis $\pa{\ket{p_1,\dots,p_N}}_{p_i \in \acs{\upa,\doa}}$, which we are going to rotate. We define $B_{\perp} \df B_x + i B_y$ and assume that $B_{\perp}(x_i) \neq 0$ for all $i \in \acs{1,\dots,N}$, otherwise the corresponding one-particle operators are already diagonal. For one particle and $B \in \R^3$, we define the rotated orthonormal spin basis
 \begin{align*}
 \ket{\ra} \df \f{\pa{ B_z + \ab{B} } \ket{\upa} + B_{\perp} \ket{\doa}}{\sqrt{ \bpa{ B_z + \ab{B} }^2 + \ab{B_{\perp}}^2}},\bhs \hs\hs\hs\hs\hs\ket{\la} \df \f{\pa{ B_z - \ab{B} } \ket{\upa} + B_{\perp} \ket{\doa}}{\sqrt{ \bpa{ B_z - \ab{B} }^2 + \ab{B_{\perp}}^2}}.
 \end{align*}
The operator $B \cdot \sigma$ is then diagonal on this one-particle basis, $\sigma \cdot B \ket{\ra} = \ab{B} \ket{\ra}$, $\sigma\cdot B \ket{\la} = - \ab{B} \ket{\la}$. Now we work at a fixed $(x_1,\dots,x_N)$ such that all $B(x_i)$ are finite. We define $B_{\ra \upa} \df B_z + \ab{B}$, $B_{\ra \doa} \df B_{\perp} =: B_{\la \doa}$, $B_{\la \upa} \df B_z - \ab{B}$, and for $N$ bodies we define similar rotations, that is for $(s_i)_{1\le i \le N} \in \acs{ \ra, \la}^N$, where $(-1)^{\ra} \df 1$ and $(-1)^{\la} \df -1$,
\begin{align*}
\ket{s_1,\dots,s_N}&  \df \f{\sum_{(p_i)_{1 \le i \le N} \in \acs{\upa,\doa}^N} \prod_{i=1}^N B_{s_i p_i}(x_i) \ket{ p_1,\dots,p_N}} {{\prod_{j=1}^N \sqrt{ \bpa{B_z(x_i) + (-1)^{s_j} \ab{B(x_i)}}^2 + \ab{B_{\perp}(x_i)}^2 }}},
 \end{align*}
which is built from $N$ consecutive one-body rotations. We finally compute
 \begin{align*}
\pa{\sum_{i=1}^N B(x_i)\cdot\sigma_i   } \ket{s_1,\dots,s_N} = \pa{ \sum_{i=1}^N (-1)^{s_i} \ab{B(x_i)}} \ket{s_1,\dots,s_N}. \hs\hs\hs\hs\hs \qed
 \end{align*}

\bibliographystyle{siam}
\bibliography{biblio}
\end{document}